\documentclass[conference]{IEEEtran}
\IEEEoverridecommandlockouts

\usepackage{fancybox}
\usepackage{mathtools}
\usepackage{nicematrix}
\usepackage{float}
\usepackage{tabularx}
\usepackage{booktabs}
\usepackage{enumitem} 

\usepackage{graphicx} 

\usepackage{color, colortbl}
\usepackage{multirow}
\usepackage[linesnumbered,ruled,vlined]{algorithm2e} 

\usepackage[utf8]{inputenc}
\usepackage[english]{babel}
\usepackage{amsthm,amssymb}
\usepackage{amsmath}
\usepackage{bm}
\usepackage{optidef}
\usepackage{mathrsfs}

\usepackage[font=small,labelfont=bf,labelsep=space]{caption}
\usepackage{xcolor}
\def\BibTeX{{\rm B\kern-.05em{\sc i\kern-.025em b}\kern-.08em
    T\kern-.1667em\lower.7ex\hbox{E}\kern-.125emX}}

\usepackage{tikz}
\usetikzlibrary{calc,patterns,decorations.pathreplacing,decorations.pathmorphing,matrix,positioning,arrows.meta,fit,bending}
\theoremstyle{definition}
\newtheorem{example}{Example}

\newtheorem{proposition}{Proposition}
\newtheorem{lemma}{Lemma}
\newtheorem{corollary}{Corollary}

\newtheorem{definition}{Definition}

\newtheorem{remark}{Remark}

\DeclareMathOperator{\supp}{supp}

\DeclareMathOperator{\LTA}{LTA}

\DeclareMathOperator{\w}{w}

\newcommand{\wt}{\operatorname{wt}}
\newcommand{\wm}{\w_{\min}}

\newcommand{\I}{\mathcal{I}}

\renewcommand{\S}{\mathcal{S}}

\newcommand{\J}{\mathcal{J}}

\newcommand{\C}{\mathcal{C}}

\newcommand{\bB}{\mathbf{B}}

\newcommand{\ind}{\operatorname{ind}}
\newcommand{\ev}{\operatorname{ev}}
\newcommand{\Stab}{\operatorname{Stab}}

\newcommand{\ft}{\mathbb{F}_2}

\newcommand{\Alow}{{\rm LTA}(m,2)}

\newcommand{\Mon}{\mathcal{M}_{m}}

\begin{document}

\title{Generalized Weight Structure of Polar Codes: Selected Template Polynomials}

\author{
\IEEEauthorblockN{Mohammad Rowshan}
\IEEEauthorblockA{University of New South Wales\\
Sydney, Australia\\
Email: m.rowshan@unsw.edu.au}
\and
\IEEEauthorblockN{Vlad-Florin Dr\u{a}goi}
\IEEEauthorblockA{Aurel Vlaicu University\\
Arad, Romania\\
Email: vlad.dragoi@uav.ro}
}

\maketitle
\pagestyle{empty}

\begin{abstract}
Polar codes can be viewed as decreasing monomial codes, revealing a rich algebraic structure governed by the lower-triangular affine (LTA) group. We develop a general framework to compute the Hamming weight of codewords generated by sums of monomials, express these weights in a canonical dyadic form, and derive closed expressions for key structural templates (disjoint sums, nested blocks, complementary flips) that generate the low and intermediate weight spectrum. Combining these templates with the LTA group action, we obtain explicit multiplicity formulas, yielding a unified algebraic method to characterize and enumerate codewords. 
\end{abstract}


\section{Introduction}
Polar codes, introduced by Ar{\i}kan~\cite{arikan}, form the first explicit family of codes that provably achieve the symmetric capacity of binary-input memoryless channels. They were adopted in 3GPP 5G NR for control channels~\cite{3GPP}, which triggered a large body of work on constructions and decoding algorithms. In contrast, comparatively less attention has been devoted to their algebraic structure and, in particular, to the distance spectrum and exact enumeration of low-weight codewords.

The weight distribution directly governs the block error probability under maximum-likelihood (ML) decoding via the union bound~\cite[Sect.~10.1]{lin_costello}. The leading terms involve the minimum distance and the multiplicities of minimum and near-minimum weight codewords (the \emph{error coefficient}). Even for polar codes, only partial answers are known.

A major step forward was recognizing that polar and Reed–Muller codes are both \emph{decreasing monomial codes}~\cite{bardet2016crypt,bardet}, which exposes a large permutation group—the lower triangular affine (LTA) group—used to characterize minimum-weight codewords~\cite{bardet}, accelerate decoding via permutations~\cite{GEECB21}, and study automorphisms. Closed-form enumerators are now known for minimum-weight codewords~\cite{bardet,rowshan2023formation}, for weights in $[1.5\wm,2\wm)$~\cite{vlad1.5d,ye2024distribution,rowshan2024weight,dragoi2025weight}, and for weight $2\wm$~\cite{rowshan2024weight}, mainly for low- and high-rate regimes where the full weight distribution can be determined; see also related enumeration methods in~\cite{yao,ellouze2023low,liu2025parity}.

In this paper, we advance the algebraic analysis of decreasing monomial codes by developing a general multiscale framework for their weight structure. We derive a general inclusion-exclusion formula for the Hamming weight of polynomials of the form $P(x)=h(x) \sum_i f_i(x)$ and introduce a dyadic decomposition of the normalized weight that explicitly characterizes fractional corrections $d / 2^{\ell}$ arising from monomial overlaps. Building on this, we establish closed-form weight expressions for key structural templates and utilize LTA group orbits to formulate exact counting multiplicities. This purely algebraic framework is applicable to any decreasing monomial code, including standard polar codes and Reed--Muller codes. 


\section{Monomial Representation and Decreasing Sets}
\label{sec:prelim}

We work over $\ft$ with modulo two addition. For $N=2^m$, a codeword is a vector $\bm c\in\ft^N$, with Hamming weight $\wt(\bm c)$ and support $\supp(\bm c)$.

\subsection{Boolean functions and monomials}

Let $x_0,\dots,x_{m-1}$ be Boolean variables and consider the quotient ring $\mathbb{R}_m \triangleq \ft[x_0,\dots,x_{m-1}]/\left\langle x_i^2-x_i\right\rangle_{i=0}^{m-1},$ so a polynomial is a Boolean function on $\{0,1\}^m$. A monomial is
\[
f = \prod_{j=0}^{m-1} x_j^{i_j},\qquad i_j\in\{0,1\}.
\]
Let $\Mon$ be the set of all monomials. If $f = x_{l_1}\cdots x_{l_s}$ with $l_1<\dots<l_s$, define 
$\ind(f) \triangleq \{l_1,\dots,l_s\},\qquad
\deg(f)\triangleq s$. 
The \emph{evaluation map} $\mathrm{ev}:\mathbb{R}_m\to\ft^{2^m}$ evaluates a polynomial at all points of $\{0,1\}^m$ in some fixed order; it is a vector-space isomorphism. For a Boolean function $f$ we write
\[
|f| \triangleq |\{x\in\{0,1\}^m : f(x)=1\}| = \wt(\mathrm{ev}(f)).
\]

If $f$ is a monomial of degree $t$, its support consists of all points where its $t$ variables are $1$ and the remaining $m-t$ variables are free; hence
\begin{equation}
\wt(\mathrm{ev}(f)) = 2^{m-\deg(f)}.
\label{eq:monom-weight}
\end{equation}

\subsection{Monomial codes and decreasing sets}

\begin{definition}[Monomial code]
Let $\I\subseteq\Mon$. The monomial code generated by $\I$ is
\[
\C(\I)\triangleq\operatorname{span}\{\mathrm{ev}(f)\mid f\in\I\}\subseteq\ft^{2^m}.
\]
\end{definition}

Let $r^+(\C)\triangleq\max_{f\in\I}\deg(f)$. Using \eqref{eq:monom-weight}, the minimum distance is $d_{\min}(\C(\I)) = 2^{m-r^+(\C)}$ \cite{bardet}.

We use the partial order $f\mid g$ if $\ind(f)\subseteq\ind(g)$, and the decreasing order $\preceq$ from~\cite{bardet2016crypt,bardet}. 

\begin{definition}[Decreasing set and decreasing monomial code]
A set $\I\subseteq\Mon$ is \emph{decreasing} if $f\in\I$ and $g\preceq f$ imply $g\in\I$. A decreasing monomial code is a monomial code $\C(\I)$ with $\I$ decreasing.
\end{definition}

Reed--Muller and (binary) polar codes are examples of decreasing monomial codes~\cite{bardet2016crypt}.

\subsection{LTA group and orbits}

Let $(\bB,\bm\varepsilon)$ be an affine map on $\ft^m$ with $\bB\in\mathrm{GL}(m,\ft)$ and $\bm\varepsilon\in\ft^m$. The \emph{lower triangular affine group} is
\[
\mathrm{LTA}(m,2)\triangleq\big\{(\bB,\bm\varepsilon) : \bB\ \text{lower triangular},\ b_{ii}=1\big\}.
\]
Each $(\bB,\bm\varepsilon)$ acts on the variables as
\[
x_i \mapsto y_i = x_i + \sum_{j<i} b_{ij} x_j + \varepsilon_i,
\]
and on a monomial $f=\prod_{i\in\ind(f)} x_i$ by substitution:
\[
(\bB,\bm\varepsilon)\cdot f \triangleq \prod_{i\in\ind(f)} y_i.
\]

This action permutes coordinates of $\mathrm{ev}(f)$, hence preserves Hamming weight. We write the orbit
\[
\mathrm{LTA}(m,2)\cdot f \triangleq \{(\bB,\bm\varepsilon)\cdot f\}.
\]

For decreasing monomial codes, $\mathrm{LTA}(m,2)$ is a subgroup of the automorphism group of $\C(\I)$, and orbits of maximal-degree monomials generate all minimum-weight codewords~\cite{bardet}.


\begin{definition}[Partition weights $|\lambda_f|$ and $|\lambda_f (g)|$]
Let $f\!=\!x_{i_0}\cdots x_{i_{s-1}}$ be a monomial with
$0\!\le\! i_0\!<\!\dots\!<\!i_{s-1}\!\le\! m-1$.
Set $\lambda_f \!=\! \bigl(i_{s-1}-(s-1),\,i_{s-2}-(s-2),\,\dots,\,i_0-0\bigr)$, then
\[
|\lambda_f| \;=\; \sum_{t=0}^{s-1} (i_t-t).
\]

More generally, for another monomial
$g=x_{j_0}\cdots x_{j_{t-1}}$ with $0\le j_0<\dots<j_{t-1}\le m-1$,
define
\[
|\lambda_f (g)|
\;\triangleq\;
\sum_{\ell=0}^{t-1}
\Bigl(j_\ell - 
\bigl|\{u\in\ind(f)\cup\{j_0,\dots,j_{\ell-1}\} : u<j_\ell\}\bigr|\Bigr).
\]
Equivalently, for each $j\in\ind(g)$, we count the number of indices
below $j$ that are not already occupied by $\ind(f)$ or by earlier
variables of $g$.

In particular, for $g=f$ this reduces to $|\lambda_f f|=|\lambda_f|$,
and the LTA-orbit of $f$ satisfies
\begin{equation}\label{eq:orbit-single-again}
    \bigl|\LTA(m,2)\cdot f\bigr| = 2^{\deg(f)+|\lambda_f|}.
\end{equation}

\end{definition}

\section{A General Weight Formula for Polynomials}
\label{sec:general-weight}

We now derive a general formula for the Hamming weight of polynomials built as a product of a common factor and a sum of residual monomials. This will be our main tool.


Let $\mathcal{P}_r$ be the set of polynomials in $\mathbb{R}_m$ of total degree $r$. For $P\in\mathcal{P}_r$, we write it in the form
\begin{equation}
P(x) = h(x)\cdot\Big(\sum_{i=1}^q f_i(x)\Big),
\label{eq:P-factorised}
\end{equation}
where $h$ is a monomial (the common factor) and the $f_i$ are residual monomials. Let $a_{\max}\triangleq\max_i \deg(f_i)$ and
\begin{equation}\label{eq:deg-h}
\deg(h)=r-a_{\max}.
\end{equation}
This representation always exists (allowing $h=1$). Let
\[
F(x) \triangleq \sum_{i=1}^q f_i(x).
\]

\noindent We normalize weights with respect to degree $r$ as
\[
d \triangleq 2^{m-r},\qquad
\Sigma(P) \triangleq \frac{\wt(P)}{d}.
\]


\noindent For $\varnothing\neq\mathcal S\subseteq\{1,\dots,q\}$, define the union degree
\[
u_{\mathcal S}\triangleq\deg\Big(\operatorname{lcm}\{f_i : i\in\mathcal S\}\Big).
\]

\begin{proposition}[General weight formula]
\label{prop:general-weight}
Let $P(x)=h(x)\sum_{i=1}^q f_i(x)$ be as in~\eqref{eq:P-factorised}, with $\deg(P)=r$ and $d=2^{m-r}$. Then
\begin{equation}
\wt(P) = d\,\Sigma(F),
\end{equation}
where
\begin{equation}
\boxed{
\Sigma(F) \triangleq
\sum_{\varnothing\neq\mathcal S\subseteq\{1,\dots,q\}}
(-2)^{|\mathcal S|-1}\,2^{\,a_{\max}-u_{\mathcal S}}.
}
\label{eq:Sigma-F}
\end{equation}
\end{proposition}


\noindent\textbf{Kernel-based description of weight templates.}
By Proposition~\ref{prop:general-weight}, once the residual monomials $\{f_i\}$ (the kernel $Q$ of
maximum degree $a_{\max}$) are fixed, the \emph{normalized weight} $\Sigma(F)$ is fully determined by
the support combinatorics; the head $h$ only raises the ambient degree from $a_{\max}$ to any
$r \ge a_{\max}$ via $\deg(h)=r-a_{\max}$, without affecting $\Sigma(F)$. Thus, for any target
normalized weight $\Sigma^\star$ one can:
\begin{enumerate}[label=(\roman*)]
  \item enumerate kernels $Q=\sum_i f_i$ with $\Sigma(F)=\Sigma^\star$ (via the union degrees $u_{\mathcal S}$), and
  \item for each such kernel, generate all degree-$r$ templates $P=hQ$ by choosing monomials $h$ of degree $r-a_{\max}$ whose supports are disjoint from that of $Q$.
\end{enumerate}
When classifying templates, we also quotient by affine equivalence: kernels $Q,Q'$ related by an
element of $\LTA(m,2)$ yield the same family of patterns up to code automorphisms. In practice, one
therefore works with canonical kernel representatives modulo $\LTA(m,2)$ and then attaches all
admissible heads $h$ (including $h=1$) for ambient degrees $r \ge a_{\max}$.


\begin{lemma}[Dyadic decomposition]
\label{lem:dyadic}
Let $\Sigma(F)$ be given by~\eqref{eq:Sigma-F}. Define
\(
U \triangleq \max_{\varnothing\neq\mathcal S} u_{\mathcal S}, 
k \triangleq U-a_{\max}. 
\)
Then
\[
\Sigma(F) = \frac{N}{2^k}
= \sum_{j=j_{\min}}^{j_{\max}} b_j 2^{\,j-k},
\]
where for $N\!=\!\sum_j b_j 2^j, b_j\!\in\!\{0,1\}$, we have
\[
N \triangleq 2^{U-a_{\max}}\,\Sigma(F)
= \sum_{\varnothing\neq\mathcal S} (-1)^{|\mathcal S|-1}\,2^{\,U-u_{\mathcal S}+|\mathcal S|-1}
\in\mathbb{Z}.
\]
\end{lemma}

Thus, every 
$\Sigma(P)$ 
is a dyadic rational of the form
\[
\Sigma(P) = \sum_{\ell=\ell_{\min}}^{\ell_{\max}} \frac{c_\ell}{2^\ell},
\]
with integer coefficients $c_\ell$. In practice, for small $q$ one observes that all $c_\ell$ are odd, and different overlap patterns generate characteristic dyadic signatures.

\section{Structural Templates for 
Weights}
\label{sec:templates}

We now instantiate Proposition~\ref{prop:general-weight} for several structural templates that occur naturally in decreasing monomial codes. In the next section these templates will be combined with the algebraic structure of $\C(\I)$ to obtain enumeration formulas.


\begin{lemma}[Disjoint $k$-sum]
\label{lem:disjoint-k}
Let $f_1,\dots,f_k$ be monomials of degree $r$ with pairwise disjoint supports and $kr\le m$. Let
\begin{equation}\label{eq:disjoint-k}
    P(x) = \sum_{i=1}^k f_i(x),\qquad r=\deg(P),
\end{equation}

and $d=2^{m-r}$. Then
\begin{equation}\label{eq:weight-disjoint-k}
    \wt(P) = 2^{m-kr-1}\Big(2^{kr} - (2^r-2)^k\Big),
\end{equation}
and the normalized weight is
\begin{equation}
\Sigma(P) =
\frac{\wt(P)}{d}
= 2^{r-1}\Big(1 - (1-2^{1-r})^k\Big).
\label{eq:Sigma-k-sum}
\end{equation}
\end{lemma}

\begin{proof}

Apply Proposition~\ref{prop:general-weight} with $h=1$, $q=k$, $a_{\max}=r$ and note that $u_{\mathcal S}=r|\mathcal S|$ by disjointness. Grouping terms by $|\mathcal S|$ yields~\eqref{eq:Sigma-k-sum}.
\end{proof}


\begin{remark}
\label{rem:dyadic-2-plus}
For $\wt(P)$ to be of the dyadic form
\[
  \wt(P) = d\Bigl(2 + \sum_{j\ge 2}\frac{1}{2^j}\Bigr),
\]
which forces $2 \le \Sigma_k(r) < 2.5$, one checks (using
\eqref{eq:Sigma-k-sum}) that the only possible case is
\[
  r=3,\quad k=3,\quad
  \Sigma_3(3) = 2 + \frac14 + \frac1{16},
\]
i.e.,
\[
  \wt(P) = 2^{m-3}\Bigl(2 + \frac14 + \frac1{16}\Bigr).
\]

On the other hand, \eqref{eq:Sigma-k-sum} also describes precisely when
$\wt(P)<2d$, i.e.\ $\Sigma_k(r)<2$. Writing
\[
  \Sigma_k(r) = 2^{r-1}\Bigl(1-(1-2^{1-r})^k\Bigr),
\]
we get
\[
  \Sigma_1(r) = 1,\qquad
  \Sigma_2(r) = 2 - 2^{1-r}
               = 1 + \sum_{j=1}^{r-1}\frac{1}{2^j}
               < 2.
\]
For $r=2$, the formula simplifies to
\[
  \Sigma_k(2) = 2\bigl(1-(1/2)^k\bigr)
              = 2 - 2^{1-k}
              = 1 + \sum_{j=1}^{k-1}\frac{1}{2^j} < 2.
\]
\end{remark}

\begin{example}
For $r=3$ and $k=3$ disjoint cubics (e.g.\ $X_1X_2X_3$, $X_4X_5X_6$, $X_7X_8X_9$),
\[
\Sigma(P) = 2^{2}\big(1-(1-2^{-2})^3\big)
= 4\big(1-(3/4)^3\big)
= \tfrac{37}{16}.
\]
Thus $\wt(P) = (37/16)\,d$ with $d=2^{m-3}$.
\end{example}

\begin{lemma}[rank-$\ell$ and a degree drop]
\label{lem:affine-deg-r}
Let $r\ge2$ and $\ell\ge0$. Let $f_1,\dots,f_\ell,g\in\ft[X_0,\dots,X_{m-1}]$
be monomials with
\[
\deg(f_j)=r\ (1\le j\le\ell),\qquad
\deg(g)=r-1,
\]
and
\(
\ind(f_j)\cap\ind(f_{j'})=\emptyset\ (j\neq j'),\;
\ind(g)\cap\ind(f_j)=\emptyset\ (1\le j\le\ell).
\)
Set
\begin{equation}\label{eq:rank-l-deg-drop}
P(X)=\sum_{j=1}^{\ell} f_j(X)+g(X)
\end{equation}
and fix ambient degree $r$, so $d=2^{m-r}$. Then
\[
\wt(P)
=
2^{m-1}
\Bigl(
  1 - \bigl(1-2^{1-r}\bigr)^{\ell}\bigl(1-2^{2-r}\bigr)
\Bigr)
\Bigr).
\]
In particular, for $\ell=0$ one has $\wt(P)=2^{m-r+1}=2d$ (a single
monomial of degree $r-1$), and for $r=2$ one has $\wt(P)=2^{m-1}=2d$
for all $\ell\ge0$.
\end{lemma}


\begin{lemma}[Complementary flip weight]
\label{lem:compl-flip}
Let $f,g\in\mathcal R_m$ be monomials with $\deg(f)=r$ and $\deg(g)=s$, and let
$j\in\{0,\dots,m-1\}$ satisfy
\[
\ind(g)\cap\bigl(\ind(f)\cup\{j\}\bigr)=\varnothing.
\]
Define
\[
P(X) \;=\; f(X) + (X_j+1)\,g(X),
\]
and set $d \triangleq 2^{m-r}$. Then
\begin{equation}
\label{eq:compl-flip-weight}
\wt\bigl(\ev(P)\bigr) \;=\; d\bigl(1+2^{r-s-1}\bigr).
\end{equation}
In particular, if $s=r$ then $\wt(\ev(P))=\tfrac32 d$, and if $s=r-1$ then
$\wt(\ev(P))=2d$.
\end{lemma}



\begin{lemma}[Shared 3-term kernels]
\label{lem:shared-3term-weight}
Let $r\ge3$ and let $h$ be a monomial of degree $\deg(h)=r-3$ whose
variables are disjoint from $\{X_1,\dots,X_7\}$.  Set $d=2^{m-r}$ and
ambient degree $r=\deg(h)+3$. 
Define
\[
P_B(X)
  = h(X)\Bigl(
      \underbrace{X_1X_2X_3}_{f_1}
    + \underbrace{X_2X_4X_5}_{f_2}
    + \underbrace{X_3X_4X_6}_{f_3}
    \Bigr),
\]
\[
P_C(X)
  = h(X)\Bigl(
      \underbrace{X_1X_2X_3}_{g_1}
    + \underbrace{X_3X_4X_5}_{g_2}
    + \underbrace{X_4X_6X_7}_{g_3}
    \Bigr),
\]
and set $Q_B=f_1+f_2+f_3$, $Q_C=g_1+g_2+g_3$. Then
\[
\wt(P_B)=\wt(P_C)=2^{m-r+1}=2d,
\quad
\Sigma(P_B)=\Sigma(P_C)=2.
\]
\end{lemma}



\subsection{Nesting (weight scaling)}

\begin{lemma}[Nesting / weight-scaling]
\label{lem:nesting}
Let $h$ be a monomial of degree $s$ and $Q$ a polynomial whose variables are disjoint from those of $h$, with maximum monomial degree $t$. Fix an ambient degree $r\ge s+t$ and define
\[
P = h\cdot Q.
\]
Let $d_Q=2^{m-t}$ and $d=2^{m-r}$, and write
\[
\Sigma(Q) \triangleq \frac{\wt(Q)}{d_Q},\qquad
\Sigma(P) \triangleq \frac{\wt(P)}{d}.
\]
Then
\begin{equation}
\Sigma(P) = 2^{\,r-(s+t)}\,\Sigma(Q).
\label{eq:nesting}
\end{equation}
In particular, if $r=s+t$ then $\Sigma(P)=\Sigma(Q)$.
\end{lemma}

\begin{proof}
Because $h$ and $Q$ are variable-disjoint, $hQ$ is $1$ exactly when $h=1$ and $Q=1$. Thus $\wt(P)=\wt(Q)/2^s$. Normalizing,
\[
\Sigma(P) = \frac{\wt(Q)/2^s}{2^{m-r}}
= \Sigma(Q)\cdot 2^{\,r-(s+t)}.
\]
\end{proof}




\begin{corollary}[Nesting of $2d$ kernels]
\label{cor:nested-2d}
Let $Q(X)$ be a polynomial of maximum
degree $t\ge1$ such that
\[
\wt(\ev(Q)) = 2^{m-t+1},
\qquad\text{i.e.}\qquad
\Sigma(Q)\triangleq\frac{\wt(Q)}{2^{m-t}}=2.
\]
This applies in the following two cases (among others):
\begin{enumerate}
  \item $Q$ is a rank-$\ell$ degree–drop kernel as in
        Lemma~\textup{\ref{lem:affine-deg-r}} (with its parameter
        $r=t$) and either $\ell=0$ or $t=2$; then
        Lemma~\textup{\ref{lem:affine-deg-r}} gives
        $\wt(\ev(Q))=2^{m-t+1}=2\cdot 2^{m-t}$.
  \item $Q$ is a complementary–flip kernel as in
        Lemma~\textup{\ref{lem:compl-flip}} with
        $\deg(f)=t$ and $\deg(g)=t-1$; then
        Lemma~\textup{\ref{lem:compl-flip}} gives
        $\wt(\ev(Q))=2^{m-t}+2^{m-(t-1)-1}=2^{m-t+1}
        =2\cdot 2^{m-t}$.
\end{enumerate}

Let $h$ be a monomial of degree $s$ whose variables are disjoint from
those of $Q$, and set
\[
r\triangleq s+t,\qquad d\triangleq 2^{m-r},\qquad
P(X)\triangleq h(X)\,Q(X).
\]
Then
\[
\wt\bigl(\ev(P)\bigr)=2d=2^{m-r+1}.
\]
\end{corollary}


\section{Enumeration of Templates' Codewords}
\label{sec:enumeration}

In this section, we give closed-form enumerator formulas, in the spirit of the minimum-weight expression, for the structural templates introduced earlier.

\subsection{General orbit-counting scheme for weight templates}
\label{subsec:general-counting}

All our counting formulas follow the same pattern.  We start from a \emph{seed}
\(
P(X) = h(X)\,Q(X),
\)
where $h$ is a monomial (\emph{head}) and $Q$ is a polynomial
(\emph{kernel}).  Write
\[
Q(X) = Q_{\max}(X) + R(X),
\]
where $Q_{\max}$ is the sum of all highest-degree monomials of $Q$ and
$R$ collects the lower-degree terms:
\[
Q_{\max}(X) = \sum_{i=1}^{\mu} t_i(X),\quad
\deg(t_i)=r_{\max},\quad
\deg(R)<r_{\max},
\]
with all $t_i$ supported outside $\ind(h)$.  The degree of $P$ is
\[
r = \deg(h)+r_{\max},\qquad d=2^{m-r}.
\]

We work in the \emph{head stabilizer}
\[
G_h \;\triangleq\; \mathrm{Stab}_{\LTA(m,2)}(h)
  \;=\; \{\gamma\in\LTA(m,2) : \gamma\cdot h = h\}.
\]
Thus, $h$ is fixed but the kernel $Q$ may move under $G_h$.  For a seed
$P=hQ$ the head-fixed orbit
\[
\mathcal{O}(P) \;\triangleq\; G_h\cdot P
\]
consists of all polynomials of the form $h\,(\gamma\cdot Q)$ with
$\gamma\in G_h$. 

\paragraph{Head contribution.}
For a monomial $h$ with partition weight $|\lambda_h|$, we have
\[
|\LTA(m,2)\cdot h| = 2^{\deg(h)+|\lambda_h|},
\]
and the same exponent counts the free LTA parameters on the variables
in $\ind(h)$ inside $G_h$.
In the special case $h\equiv1$, we have $\deg(h)=|\lambda_h|=0$ and
$G_h=\LTA(m,2)$.

\paragraph{Head contribution.}
For a monomial $h$ with partition weight $|\lambda_h|$, the usual LTA
orbit formula gives
\[
|\LTA(m,2)\cdot h| = 2^{\deg(h)+|\lambda_h|}.
\]
In the special case $h\equiv1$ we have $\deg(h)=|\lambda_h|=0$ and
$G_h=\LTA(m,2)$.

\paragraph{Kernel decomposition: $Q_{\max}$ and $R$.}
The kernel $Q$ may contain lower-degree monomials $R$ that are:
\begin{itemize}[leftmargin=*, itemsep=0pt]
  \item Fixed or template-determined: 
        in many low-weight templates (e.g.\ complementary flip) $R$ is uniquely determined once the head $h$
        and the highest-degree tails $t_i$ are fixed.  In this case, $R$
        does not introduce extra parameters;  for any
        $\gamma\in G_h$, the image $\gamma\cdot R$ is a
        deterministic function of $\gamma\cdot Q_{\max}$.
  \item Independent (disjoint variables): 
        if $\ind(R)$ is disjoint from $\ind(h)\cup\bigcup_i\ind(t_i)$,
        the monomials of $R$ form an additional block of tails with their
        own LTA degrees of freedom and possible collisions.  Their
        contribution factors multiplicatively into the orbit exponent
        (see below).
  \item Partially shared:
        if some variables in $R$ are shared with $h$ or with the $t_i$,
        the tails of $R$ are simply added to the global tail set, and
        their effect is absorbed into the same $\beta_Q$ and $\alpha_Q$
        parameters defined below.
\end{itemize}

\paragraph{Tail contribution and stabilizer (collisions)}
Write the kernel as a sum of monomials
\[
Q = \sum_{j=1}^{\nu} u_j,\quad \deg(u_j)\le r_{\max},\quad
\ind(u_j)\cap\ind(h)=\varnothing.
\]
The corresponding degree-$r$ monomials in $P$ are
\[
f_j \;=\; h\,u_j,\qquad 1\le j\le\nu.
\]
Knowing that the affine action of $\LTA(m,2)$
on monomials is multiplicative \cite{rowshan2025weight}, for every tail $u_j$ and
every $\gamma\in G_h$,
\[
\gamma\cdot(h\,u_j)
  = (\gamma\cdot h)(\gamma\cdot u_j)
  = h\,(\gamma\cdot u_j).
\]
Thus, the head $h$ is fixed and the group acts only on the tails
$u_j$. In particular, the orbit of $f_j$ under $G_h$ is in bijection
with the \emph{tail–orbit}
\[
\mathcal{O}_j \;\triangleq\; G_h\cdot u_j.
\]

The degree-$r$ and lower-degree parts of $P$ generated by the tails are
described by the Minkowski sum
\[
\mathcal{S}
  \;\triangleq\;
  \mathcal{O}_1 + \cdots + \mathcal{O}_\nu
  = \{a_1+\cdots+a_\nu : a_j\in\mathcal{O}_j\}.
\]
Up to the common factor $h$, the degree-$r$ term of $\gamma\cdot P$ is $h(a_1+\cdots+a_\nu)$ for some $a_j\in\mathcal{O}_j$, so every collision among degree-$r$ monomials of $P$ is induced by a collision in $\mathcal{S}$. 
For each pair $(i,j)$ with $1\le i<j\le\nu$, we define the
\emph{pairwise collision exponent} $\alpha_{u_i,u_j}$ by \cite[Prop. 1]{vlad1.5d}
\[
|\mathcal{O}_i+\mathcal{O}_j|
  \;=\;
  \frac{|\mathcal{O}_i|\,|\mathcal{O}_j|}{2^{\alpha_{u_i,u_j}}}.
\]
Note that $\alpha_{u_i,u_j}$ may be nonzero even if $\ind(u_i)$ and
$\ind(u_j)$ are disjoint: in Type–II structures
$P=h\sum f_i$ with $\deg(h)=r-2$ and disjoint quadratic tails
$u_i=f_i/h$, the common head $h$ and the lower–triangular constraints
still create collisions (this is exactly the collision parameter in  \cite[Def. 8]{vlad1.5d}).

For a given kernel $Q$, we define the total collision exponent and, under the standard assumption that all higher–order collisions are
induced by these pairwise overlaps, we obtain
\[
\alpha_Q \;\triangleq\; \sum_{1\le i<j\le\nu} \alpha_{u_i,u_j}, \quad |\mathcal{S}|
  \;=\;
  \frac{\prod_{j=1}^{\nu} |\mathcal{O}_j|}{2^{\alpha_Q}}.
\]


\paragraph{Tail and linear degrees of freedom on the kernel.}
Let $Q=\sum_{j=1}^\nu u_j$ be written as a sum of tails $u_j$
(monomials of degree $\le r_{\max}$, disjoint from $\ind(h)$).  For each
tail $u_j$ we have its usual single-tail orbit size
$2^{\deg(u_j)+|\lambda_{u_j}|}$, so $\deg(u_j)+|\lambda_{u_j}|$ is the
number of LTA parameters we would obtain if $u_j$ were acting alone.

When several tails are present simultaneously, the admissible
lower--triangular coefficients form a linear subspace $V_Q$ of the
ambient parameter space.  Its dimension exceeds
$\sum_j (\deg(u_j)+|\lambda_{u_j}|)$ by a nonnegative “mixing” term
$\beta_Q^{\rm mix}$, which accounts for
\begin{itemize}
  \item rows shared by several tails (variables appearing in more than
        one $u_j$), and
  \item rows belonging only to lower-degree tails (e.g.\ linear factors)
        that introduce additional joint degrees of freedom.
\end{itemize}
We therefore write
\[
\beta_Q
  \;\triangleq\;
  \sum_{j=1}^\nu \bigl(\deg(u_j)+|\lambda_{u_j}|\bigr)
  \;+\;
  \beta_Q^{\rm mix},
\]
so that $2^{\beta_Q}$ is exactly the number of kernel-side LTA degrees
of freedom that preserve the template $Q$.

\paragraph{Master orbit-size formula.}
Let $G_h = \mathrm{Stab}_{\LTA(m,2)}(h)$ be the head stabilizer and
$\mathcal{O}(P)=G_h\cdot P$ the head-fixed orbit of the seed
$P=hQ$.  The admissible LTA parameters in $G_h$ that send $P$ to a
polynomial with the same head $h$ and the same kernel template $Q$
form an affine space of size
\[
2^{\deg(h)+|\lambda_h|+\beta_Q},
\]
where $\deg(h)+|\lambda_h|$ counts the head-side degrees of freedom and
$\beta_Q$ the kernel-side ones.  Collisions among these parameters are
accounted for by the collision exponent $\alpha_Q$ (defined via the
Minkowski sum of tail–orbits), so the orbit size is
\[
\boxed{
\bigl|G_h\cdot P\bigr|
  \;=\;
  2^{\,\deg(h)+|\lambda_h|+\beta_Q-\alpha_Q}.
}
\]
The special case
$h\equiv1$ is obtained by setting $\deg(h)=|\lambda_h|=0$ and
$G_h=\LTA(m,2)$.

\subsection{Enumeration/counting formulae}
In this template, we have $h=1$.
\begin{proposition}[Enumeration of disjoint $k$-sum codewords]
\label{prop:k-disjoint-enum-decreasing-final}
Let $\C(\I)$ be a decreasing monomial code and $w_{k,r}$ as in
\eqref{eq:weight-disjoint-k}.  Let $\mathcal F_{k,r}(\I)$ be the set of
ordered $k$-tuples $F=(f_1,\dots,f_k)$ of degree-$r$ monomials in $\I$
with pairwise disjoint supports $\ind(f_i)\cap\ind(f_j)=\varnothing$
for $i\neq j$. Then the number of weight-$w_{k,r}$ codewords arising
from disjoint $k$-sums is
\begin{equation}
  \label{eq:A-k-disjoint-decreasing-final}
  A_{w_{k,r}}^{(k)}(\I)
\;=\;
\sum_{\mathclap{F\in\mathcal F_{k,r}(\I)}}
  \prod_{i=1}^k \bigl|\Alow_{f_i}\cdot f_i\bigr|
\;=\;
\sum_{\mathclap{F\in\mathcal F_{k,r}(\I)}}
  2^{\,k r + \sum_{i=1}^k |\lambda_{f_i}|}.
\end{equation}
\end{proposition}

Here, the stabilizer of each tuple $F$ is trivial, so different
tuples have disjoint orbits, and there is no collision; $\alpha_Q=0$.




\begin{proposition}[Enumeration of nested rank-$\ell$ degree–drop $2d$ codewords]
\label{prop:nested-rank-l-deg-drop-enum}
Let $\C(\I)$ be a decreasing monomial code of length $2^m$ with
maximum degree $r\ge2$ and $d=2^{m-r}$.

A \emph{nested rank-$\ell$ degree–drop seed} is a triple $(h;S;x_j)$ with
\[
P(X)=h(X)\,Q(X),\qquad
Q(X)=\sum_{f\in S}\frac{f(X)}{h(X)} + X_j,
\]
where
\(
S\subseteq \I_r,\ |S|=\ell\ge0, 
h\in\I_{r-2},\quad
\gcd(X_j,f/h)=1\ \forall\,f\in S, P\in\operatorname{span}(\I),
\)
and, when $\ell\ge2$, $\gcd(f_a,f_b)=h$ for all distinct $f_a,f_b\in S$.
(For $\ell=0$ we have $S=\varnothing$ and $Q=X_j$.)

Write $S=\{f_1,\dots,f_\ell\}$ and set $q_i\triangleq f_i/h$ (so
$\deg(q_i)=2$).  For such a seed define
\[
\begin{aligned}
E(h;S;x_j)
&\;=\;
(r-2)+2\ell+|\lambda_h|
+\sum_{i=1}^{\ell}\bigl|\lambda_{f_i}(q_i)\bigr|
\\[-0.25em]
&\qquad
-\sum_{\mathclap{1\le a<b\le\ell}}\alpha_{q_a,q_b}
\;+\;1+|\lambda_{f_1\cdots f_{\ell}}(X_j)|,
\end{aligned}
\]
with the conventions for $\ell=0$ that all sums are $0$ and
$f_1\cdots f_\ell=1$ in the linear term
$|\lambda_{f_1\cdots f_{\ell}}(X_j)|$ (which then reduces to the 
head–constrained partition weight of $X_j$).  
Then, the $\LTA(m,2)$–orbit of $P$ has size
        \[
          \bigl|\LTA(m,2)\cdot P\bigr|
          \;=\;
          2^{E(h;S;x_j)},
        \]
and distinct seeds $(h;S;x_j)$ give disjoint $\LTA(m,2)$–orbits. 
        Consequently, the total number of weight–$2d_{\min}$ codewords
        of nested rank–$\ell$ degree–drop type in $\C(\I)$ is
        \begin{equation}
          \label{eq:nested-rank-l-enum}
          A^{\mathrm{nest-dp}}_{2d_{\min}}(\I)
          \;=\;
          \sum_{(h;S;x_j)} 2^{E(h;S;x_j)},
        \end{equation}
        where the sum runs over all seeds with $\ell=|S|\ge0$ in $\C(\I)$.
\end{proposition}

This proposition is generalization of \cite[Lemma 1]{kasami1970weight}.

\clearpage
\bibliographystyle{IEEEtran}
\bibliography{refs}

\clearpage
\appendices

\section{Inclusion-Exclusion Formula for Hamming Weight of Codewords}

\begin{lemma}\label{lem:pie}
Let $G$ be a $k \times n$ generator matrix for a linear code over $\mathbb{F}_2$, with rows $g_1, g_2, \dots, g_k \in \mathbb{F}_2^n$. Let $\mathcal{J} \subseteq [1,k]$ and $\mathbf{c} = \sum_{j \in \mathcal{J}} g_j$ (addition in $\mathbb{F}_2^n$) be the codeword obtained as the sum of the rows indexed by $\mathcal{J}$. The Hamming weight $\mathrm{wt}(\mathbf{c})$ of $\mathbf{c}$ is given by
\begin{equation}\label{eq:pie_codeword}
\boxed{\;\mathrm{wt}(\mathbf{c}) = \sum_{\emptyset \neq \mathcal{S} \subseteq \mathcal{J}} (-2)^{|\mathcal{S}|-1} \left| \bigcap_{j \in \mathcal{S}} \mathrm{supp}(g_j) \right|,\;}
\end{equation}
where $\mathrm{supp}(g_i) = \{ j \in \{1,2,\dots,n\} \mid g_{i,j} = 1 \}$ is the support of row $g_i$.
\end{lemma}

\begin{proof}
Let $A_j = \mathrm{supp}(g_j)$ for $j \in \mathcal{J}$. The Hamming weight $\mathrm{wt}(\mathbf{c})$ counts positions $j \in \{1,2,\dots,n\}$ where $\mathbf{c}_j = 1$, i.e., where an odd number of rows $g_i$ for $i \in \mathcal{J}$ have $g_{i,j} = 1$. Define $d(j) = \sum_{i \in \mathcal{J}} g_{i,j}$ (sum as integers), so $\mathrm{wt}(\mathbf{c}) = \sum_{j=1}^n \mathbf{1}_{\{d(j) \text{ is odd}\}}(j)$. Using the parity indicator $\mathbf{1}_{\{d(j) \text{ is odd}\}} = \frac{1 - (-1)^{d(j)}}{2}$, we get
\[
\mathrm{wt}(\mathbf{c}) = \sum_{j=1}^n \frac{1 - (-1)^{d(j)}}{2} = \frac{1}{2} \left( n - \sum_{j=1}^n (-1)^{d(j)} \right).
\]
Compute $\sum_{j=1}^n (-1)^{d(j)}$:
\[
(-1)^{d(j)} = \prod_{i \in \mathcal{J}} (1 - 2 g_{i,j}),
\]
since $1 - 2 g_{i,j} = -1$ if $g_{i,j} = 1$ and $+1$ if $g_{i,j} = 0$. Thus,
\[
\sum_{j=1}^n (-1)^{d(j)} = \sum_{j=1}^n \prod_{i \in \mathcal{J}} (1 - 2 g_{i,j}).
\]
Expand the product:
\[
\prod_{i \in \mathcal{J}} (1 - 2 g_{i,j}) = \sum_{\mathcal{S} \subseteq \mathcal{J}} (-2)^{|\mathcal{S}|} \prod_{i \in \mathcal{S}} g_{i,j},
\]
so
\[
\begin{aligned}
\sum_{j=1}^n (-1)^{d(j)} &= \sum_{\mathcal{S} \subseteq \mathcal{J}} (-2)^{|\mathcal{S}|} \sum_{j=1}^n \prod_{i \in \mathcal{S}} g_{i,j} \\&= \sum_{\mathcal{S} \subseteq \mathcal{J}} (-2)^{|\mathcal{S}|} \left| \bigcap_{i \in \mathcal{S}} A_i \right|,
\end{aligned}
\]
since $\sum_{j=1}^n \prod_{i \in \mathcal{S}} g_{i,j} = |\bigcap_{i \in \mathcal{S}} A_i|$. The term for $\mathcal{S} = \emptyset$ is $n$. Substituting into the weight expression,
\[
\begin{aligned}
\mathrm{wt}(\mathbf{c}) &= \frac{1}{2} \left( n - n - \sum_{\mathclap{\emptyset \neq \mathcal{S} \subseteq \mathcal{J}}} (-2)^{|\mathcal{S}|} \left| \bigcap_{i \in \mathcal{S}} A_i \right| \right) \\&= \sum_{\mathclap{\emptyset \neq \mathcal{S} \subseteq \mathcal{J}}} (-1)^{|\mathcal{S}|+1} 2^{|\mathcal{S}|-1} \left| \bigcap_{i \in \mathcal{S}} A_i \right|.
\end{aligned}
\]
Since $(-1)^{|\mathcal{S}|+1} 2^{|\mathcal{S}|-1} = (-2)^{|\mathcal{S}|-1}$, we obtain
\[
\mathrm{wt}(\mathbf{c}) = \sum_{\emptyset \neq \mathcal{S} \subseteq \mathcal{J}} (-2)^{|\mathcal{S}|-1} \left| \bigcap_{j \in \mathcal{S}} A_j \right|,
\]
matching the stated formula.
\end{proof}
\begin{remark}
    The formula for $\mathrm{wt}(\mathbf{c})$ in Lemma \ref{lem:pie} is a parity-based adaptation of the principle of inclusion-exclusion (PIE) 
    \[
    \left|\bigcup_{j\in\J} A_j\right|=\sum_{\emptyset \neq \S \subseteq\J}(-1)^{|\S|+1}\left|\bigcap_{j \in \S} A_j\right|
    \]
    for finite sets $A_j = \mathrm{supp}(g_j)$, $j \in \mathcal{J}$. While standard PIE computes the size of the union $\left| \bigcup_{j \in \mathcal{J}} A_j \right|$ using alternating signs $(-1)^{|\mathcal{S}|+1}$, the given formula counts positions $j$ where an odd number of sets $A_j$ overlap, achieved by weighting intersections with $(-2)^{|\mathcal{S}|-1}$. This arises from the parity indicator $\frac{1 - (-1)^{d(j)}}{2}$, transforming the PIE sum to capture only odd-sized overlaps, aligning with the $\mathbb{F}_2$ sum in the codeword $\mathbf{c}$.
\end{remark}

\section{Proof of Proposition \ref{prop:general-weight}}
\begin{proof}
The support of $h\cdot f_i$ is the set of points where all variables in $h$ and all variables in $f_i$ are equal to 1. According to Lemma \ref{lem:pie}, we need to find
\[
\left|\bigcap_{i \in S} \operatorname{supp}\left(h\cdot f_i\right)\right|,
\]
where $\mathcal{S} \subseteq\{1, \ldots, q\}$ is a nonempty set. Given that $\operatorname{supp}\left(h\cdot f_i\right)$ is a set where all variables in $h\cdot f_i$ are 1, then $\bigcap_{i \in S} \operatorname{supp}\left(h\cdot f_i\right)$ is a set where all variables in every $h\cdot f_i$ are 1. That is, every variable appearing in any $f_i$ (for $i \in \mathcal{S}$ ) and all variables in $h$ must be 1. Thus, the intersection is the support of
\[
h \cdot \operatorname{LCM}\{f_i: i \in \S\},
\]
where LCM is the least common multiple of the monomials and collects all variables appearing in at least one. The degree of the intersection monomial is 
\[
\operatorname{deg}(h)+\operatorname{deg}\left(\operatorname{LCM}\left\{f_i: i \in \S\right\}\right).
\]
We denote the union-degree of $\S$ as
\[
u_{\S}\triangleq\operatorname{deg}\left(\operatorname{LCM}\left\{f_i: i \in \S\right\}\right),
\]
which tells you how many residual variables are ``forced to 1'' in the intersection. Knowing this and the fact that a monomial of degree $t$ has support size $2^{m-t}$, the size of intersection support is 
\[
\left|\bigcap_{i \in S} \operatorname{supp}\left(h f_i\right)\right|=2^{m-\left(\operatorname{deg} h+u_S\right)}.
\]
Now substitute $\operatorname{deg}(h)=r-a_{\text {max }}$ from \eqref{eq:deg-h}, we obtain 
\[
\left|\bigcap_{i \in \mathcal{S}} \operatorname{supp}\left(h \cdot f_i\right)\right|=2^{m-r} \cdot 2^{a_{\max }-u_S}=d \cdot 2^{a_{\max }-u_S},
\]
where $d = 2^{m-r}$. Now, using \eqref{eq:pie_codeword}, we can write 
\begin{equation*}
    \wt(P)=d \Sigma, \quad \Sigma=\sum_{\varnothing \neq \mathcal{S} \subseteq\{1, \ldots, q\}}(-2)^{|\mathcal{S}|-1} 2^{a_{\max }-u_\mathcal{S}}.
\end{equation*}
\end{proof}






\section{Proof of Lemma \ref{lem:dyadic}}

\begin{proof}
By Proposition~\ref{prop:general-weight}, the normalized weight associated
with the residual family $F=\{f_1,\dots,f_q\}$ is
\[
\Sigma(F)
=
\sum_{\varnothing\neq\mathcal S\subseteq\{1,\dots,q\}}
(-2)^{|\mathcal S|-1}\,2^{\,a_{\max}-u_{\mathcal S}}.
\]
Factor out $2^{a_{\max}-U}$:
\begin{align*}
\Sigma(F)
&=
2^{a_{\max}-U}
\sum_{\varnothing\neq\mathcal S}
(-2)^{|\mathcal S|-1}\,2^{\,U-u_{\mathcal S}}\\
&=
2^{a_{\max}-U}
\sum_{\varnothing\neq\mathcal S}
(-1)^{|\mathcal S|-1}\,2^{\,U-u_{\mathcal S}+|\mathcal S|-1}.
\end{align*}
Set
\[
N \triangleq 2^{U-a_{\max}}\Sigma(F)
=
\sum_{\varnothing\neq\mathcal S}
(-1)^{|\mathcal S|-1}\,2^{\,e_{\mathcal S}},
\]
\[
e_{\mathcal S} \triangleq U-u_{\mathcal S}+|\mathcal S|-1.
\]
For every nonempty $\mathcal S$ we have $u_{\mathcal S}\le U$ by the
definition of $U$, and $|\mathcal S|\ge1$, hence
\[
e_{\mathcal S}=U-u_{\mathcal S}+|\mathcal S|-1 \;\ge\; 0.
\]
Therefore each summand $(-1)^{|\mathcal S|-1}2^{e_{\mathcal S}}$ is an
integer, and the finite sum $N$ is an integer:
$N\in\mathbb{Z}$.

Since $k=U-a_{\max}$, we can rewrite
\[
\Sigma(F) = 2^{a_{\max}-U} N = \frac{N}{2^{U-a_{\max}}}
= \frac{N}{2^k}.
\]
Write the (finite) binary expansion of $N$ as
\[
N = \sum_{j=j_{\min}}^{j_{\max}} b_j 2^j,
\qquad b_j\in\{0,1\},
\]
where $j_{\min}$ and $j_{\max}$ are the smallest and largest indices
with $b_j=1$. Substituting into the previous identity gives
\[
\Sigma(F)
=
\frac{1}{2^k}\sum_{j=j_{\min}}^{j_{\max}} b_j 2^j
=
\sum_{j=j_{\min}}^{j_{\max}} b_j 2^{\,j-k},
\]
which is the desired dyadic decomposition.
\end{proof}

\section{Proof of Lemma \ref{lem:affine-deg-r}}

\begin{proof}
The support of $f$ has size $2^{m-r}=d$. The support of $(X_j+1)g$ consists of
the points with $X_j=0$ and all variables in $\ind(g)$ equal to $1$, so it has
size $2^{m-s-1}$. These supports are disjoint, since $f(X)=1$ forces $X_j=1$
while $(X_j+1)g(X)\neq 0$ forces $X_j=0$. Hence
\[
\wt\bigl(\ev(P)\bigr)
= 2^{m-r}+2^{m-s-1}
= d\bigl(1+2^{r-s-1}\bigr).
\]
\end{proof}

\section*{Proof of Lemma~\ref{lem:affine-deg-r}}
\begin{proof}
If $\ell=0$ then $P=g$ and
\[
\wt(P)=2^{m-(r-1)}=2^{m-r+1}=2d.
\]
On the other hand the formula gives
\[
\begin{aligned}
\wt(P)
&=
2^{m-1}\bigl(1-(1-2^{1-r})^{0}(1-2^{2-r})\bigr)\\
&=
2^{m-1}\cdot 2^{2-r}
=
2^{m-r+1},
\end{aligned}
\]
so the claim holds. Hence assume $\ell\ge1$.

Write $f_{\ell+1}=g$ and $F=\sum_{i=1}^{\ell+1} f_i=P$. Then
$\deg(P)=r$.  
For a nonempty $\mathcal S$ in \eqref{prop:general-weight} of Proposition~\ref{prop:general-weight}, let
\[
k=|\mathcal S\cap\{1,\dots,\ell\}|,\qquad
\varepsilon=\mathbf 1_{\{\ell+1\in\mathcal S\}}.
\]
By disjointness and $\deg(f_j)=r$, $\deg(g)=r-1$,
\[
u_{\mathcal S}=kr+\varepsilon(r-1),
\qquad
2^{r-u_{\mathcal S}}=2^{r-kr-\varepsilon(r-1)}.
\]

Split $\Sigma(F)=\Sigma_0+\Sigma_1$ according to $\varepsilon$.

1) $\varepsilon=0$ (subsets not containing $g$): $1\le k\le\ell$ and
\[
\begin{aligned}
\Sigma_0
&=
\sum_{\substack{\mathcal S\neq\varnothing\\\ell+1\notin\mathcal S}}
(-2)^{|\mathcal S|-1}2^{r-u_{\mathcal S}}
=
\sum_{k=1}^{\ell}\binom{\ell}{k}
(-2)^{k-1}2^{r-kr}
\\
&=
2^{r-1}\sum_{k=1}^{\ell}\binom{\ell}{k}(-2^{1-r})^k
=
2^{r-1}\bigl(1-(1-2^{1-r})^{\ell}\bigr).
\end{aligned}
\]

2) $\varepsilon=1$ (subsets containing $g$): $0\le k\le\ell$ and
\[
\begin{aligned}
\Sigma_1
&=
\sum_{\substack{\mathcal S\neq\varnothing\\\ell+1\in\mathcal S}}
(-2)^{|\mathcal S|-1}2^{r-u_{\mathcal S}}
=
\sum_{k=0}^{\ell}\binom{\ell}{k}
(-2)^k 2^{r-(kr+r-1)}\\
&=
2\sum_{k=0}^{\ell}\binom{\ell}{k}(-2^{1-r})^k
=
2(1-2^{1-r})^{\ell}.
\end{aligned}
\]

Thus
\[
\begin{aligned}
\Sigma(F)
&=
2^{r-1}\bigl(1-(1-2^{1-r})^{\ell}\bigr)
+2(1-2^{1-r})^{\ell}\\
&=
2^{r-1}
\Bigl(
  1-(1-2^{1-r})^{\ell}(1-2^{2-r})
\Bigr),
\end{aligned}
\]
and multiplying by $d=2^{m-r}$ gives the stated weight formula.
\end{proof}

\section{Proof of Corollary~\ref{lem:compl-flip}}

\begin{proof}
The support of $f$ has size $2^{m-r}=d$. The support of $(X_j+1)g$ consists of points where $X_j=0$ and all variables in $g$ are $1$, hence $| (X_j+1)g|=2^{m-s-1}$ and their supports are disjoint (since $f=1$ forces $X_j=1$). Thus
\[
\begin{aligned}
\wt(P)&=\wt(f)+\wt((X_1+1)g)=2^{m-r}+2^{m-s-1}
\\&= d\Big(1+2^{r-s-1}\Big).
\end{aligned}
\]
\end{proof}

\section{Proof of Lemma \ref{lem:shared-3term-weight}}

\begin{proof}
Apply Proposition~\ref{prop:general-weight} with $h\equiv1$, $q=3$,
$a_{\max}=3$ to $Q_B$ and $Q_C$.  Writing
\[
\Sigma(Q)
=\sum_{\varnothing\neq\mathcal S\subseteq\{1,2,3\}}
(-2)^{|\mathcal S|-1}\,2^{3-u_{\mathcal S}},
\]
where $u_{\mathcal S}=|\bigcup_{i\in\mathcal S}\ind(f_i)|$ (for $Q_B$)
or $u_{\mathcal S}=|\bigcup_{i\in\mathcal S}\ind(g_i)|$ (for $Q_C$),
we get:

- For $Q_B$:
  \[
  \ind(f_1)=\{1,2,3\},\ \ind(f_2)=\{2,4,5\},\ \ind(f_3)=\{3,4,6\},
  \]
  hence
  \[
  u_{\{1\}}=u_{\{2\}}=u_{\{3\}}=3,
  \]\[
  u_{\{1,2\}}=u_{\{1,3\}}=u_{\{2,3\}}=5,
  u_{\{1,2,3\}}=6,
  \]
  and
  \[
  \Sigma(Q_B)
  =3\cdot2^0-2\cdot3\cdot2^{-2}+4\cdot2^{-3}=2.
  \]

- For $Q_C$:
  \[
  \ind(g_1)=\{1,2,3\},\ \ind(g_2)=\{3,4,5\},\ \ind(g_3)=\{4,6,7\},
  \]
  hence
  \[
  u_{\{1\}}=u_{\{2\}}=u_{\{3\}}=3,
  \]\[
  u_{\{1,2\}}=u_{\{2,3\}}=5, 
  u_{\{1,3\}}=6, 
  u_{\{1,2,3\}}=7,
  \]
  and
  \[
  \Sigma(Q_C)
  =3\cdot2^0
   -2\bigl(2^{3-5}+2^{3-6}+2^{3-5}\bigr)
   +4\cdot2^{3-7}
  =2.
  \]

Thus, in ambient degree $3$,
\[
\wt(Q_B)=\wt(Q_C)=2^{m-3}\Sigma(Q_*)=2^{m-2}.
\]

Since $\deg(h)=r-3$ and $\ind(h)$ is disjoint from the variables of
$Q_B,Q_C$, Lemma~\ref{lem:nesting} with $s=\deg(h)$, $t=3$ and
$r=s+t=r$ gives
\[
\Sigma(P_B)=\Sigma(P_C)=\Sigma(Q_B)=\Sigma(Q_C)=2.
\]
With $d=2^{m-r}$ we obtain
\[
\wt(P_B)=\wt(P_C)=d\,\Sigma(P_*)=2^{m-r+1}=2d.
\]
\end{proof}

\section{Proof of Corollary~\ref{cor:nested-2d}}

\begin{proof}
Since $h$ and $Q$ are variable–disjoint, Lemma~\ref{lem:nesting}
applies with head degree $s$, kernel degree $t$, and ambient degree
$r=s+t$. Writing $d_Q=2^{m-t}$ and $d=2^{m-r}$, Lemma~\ref{lem:nesting}
states
\[
\Sigma(P)
\;\triangleq\;
\frac{\wt(P)}{d}
=
2^{\,r-(s+t)}\,\Sigma(Q).
\]
By assumption $\Sigma(Q)=2$ and $r=s+t$, hence
\[
\Sigma(P) = 2^{\,0}\cdot 2 = 2.
\]
Therefore
\[
\wt\bigl(\ev(P)\bigr)
=
\Sigma(P)\,d
=
2\cdot 2^{m-r}
=
2^{m-r+1}
=
2d.
\]

For case~(1), Lemma~\ref{lem:affine-deg-r} explicitly gives
$\wt(\ev(Q))=2^{m-t+1}$ when $\ell=0$ or $t=2$, so
$\Sigma(Q)=\wt(Q)/2^{m-t}=2$ and the above argument applies.

For case~(2), with $\deg(f)=t$ and $\deg(g)=t-1$,
Lemma~\ref{lem:compl-flip} yields
\[
\wt\bigl(\ev(Q)\bigr)
=
d_Q\bigl(1+2^{t-(t-1)-1}\bigr)
=
d_Q(1+1)
=
2d_Q
=
2^{m-t+1},
\]
so again $\Sigma(Q)=2$ and the same nesting argument gives
$\wt(\ev(P))=2d$.

This covers both affine degree–drop kernels and complementary–flip
kernels after nesting.
\end{proof}

\section{Proof for Proposition \ref{prop:k-disjoint-enum-decreasing-final}}

\begin{proof}
Fix $F=(f_1,\dots,f_k)\in\mathcal F_{k,r}(\I)$ and let
$S(F)=\bigcup_{i=1}^k\ind(f_i)$ be the union of the supports.
By definition, the supports $\ind(f_i)$ are pairwise disjoint, so we
may view the local LTA group $\Alow_{F}$ as the direct product
\[
  \Alow_F
  \;\cong\;
  \prod_{i=1}^k \Alow_{f_i}
\]
acting independently on each $f_i$.
For $(A_1,b_1),\dots,(A_k,b_k)\in\Alow_{f_i}$ we get a codeword of the
form
\[
  P_F^{(A_1,\dots,A_k)}
  \;=\;
  \sum_{i=1}^k (A_i,b_i)\cdot f_i.
\]
By construction of $\Alow_{f_i}$, each $(A_i,b_i)\cdot f_i$ is a
degree-$r$ monomial in $\I_r$ and their supports remain disjoint; thus
$P_F^{(A_1,\dots,A_k)}$ always has the disjoint $k$-sum shape and
weight $w_{k,r}$.

The orbit size (numerator of the Stab--orbit formula) is
\[
\begin{aligned}
|\Alow_F|
  &=
  \prod_{i=1}^k |\Alow_{f_i}\cdot f_i|
  \stackrel{\eqref{eq:orbit-single-again}}=
  \prod_{i=1}^k 2^{\deg(f_i)+|\lambda_{f_i}|}
  \\&=
  2^{\,k r + \sum_{i=1}^k|\lambda_{f_i}|}.
\end{aligned}
\]

We now show that the local stabiliser is trivial.

\paragraph*{Stabiliser analysis.}
Suppose $(A_1,b_1),\dots,(A_k,b_k)$ and
$(A_1',b_1'),\dots,(A_k',b_k')$ produce the same polynomial:
\[
  \sum_{i=1}^k (A_i,b_i)\cdot f_i
  \;=\;
  \sum_{i=1}^k (A_i',b_i')\cdot f_i.
\]
Since monomials are linearly independent and the supports of the
$f_i$’s are disjoint, the expansion into monomials with support
contained in each $\ind(f_i)$ is unique.  Hence for each fixed $i$
we must have
\[
  (A_i,b_i)\cdot f_i = (A_i',b_i')\cdot f_i
\]
as monomials.  But the orbit $\Alow_{f_i}\cdot f_i$ acts freely on
$f_i$ (see the discussion in Section~\ref{sec:prelim}), so
this implies $(A_i,b_i)=(A_i',b_i')$ for all $i$.
Therefore the only element in $\Alow_F$ that fixes $P_F$ is the
identity, i.e.
\[
  \Stab_{\Alow}(F) = \{\mathrm{id}\},
  \qquad
  |\Stab_{\Alow}(F)|=1.
\]

Thus the orbit size for $F$ is exactly $|\Alow_F|$, and summing this
over all templates $F\in\mathcal F_{k,r}(\I)$ yields
\eqref{eq:A-k-disjoint-decreasing-final}.
\end{proof}

\section{Proof for Proposition \ref{prop:nested-rank-l-deg-drop-enum}}

\begin{proof}

\emph{Orbit size for fixed seed.}
Write $S=\{f_1,\dots,f_\ell\}$ and $q_i=f_i/h$.  The degree–$r$ part of
$P$ is the Type–II kernel \cite[Thm.~2]
{rowshan2025weight}
\[
h\sum_{i=1}^{\ell}q_i,
\]
and by the general Type–II orbit formula
\cite[Prop.~3, Thm.~5]{rowshan2025weight} its $\LTA(m,2)$–orbit has
size
\[
2^{(r-2)+2\ell+|\lambda_h|
   +\sum_{i=1}^{\ell}|\lambda_{f_i}(q_i)|
   -\sum_{1\le a<b\le\ell}\alpha_{q_a,q_b}}.
\]
(For $\ell=0$, this reduces to $2^{(r-2)+|\lambda_h|}$, the usual
head–orbit size.)  Restricting to the stabilizer of the monomial
tuple~$S$ and letting it act on $X_j$, one gets
\[
\bigl|\LTA(m,2)_S\cdot X_j\bigr|
=
2^{1+|\lambda_{f_1\cdots f_{\ell}}(X_j)|},
\]
exactly as in \cite[Lemma~5, Thm.~6]{rowshan2025weight}; adding the
linear orbit creates no new collisions beyond the Type–II ones already
encoded by the $\alpha_{q_a,q_b}$.  The full orbit of $P$ thus factors
as a Cartesian product of the Type–II monomial orbit and this linear
orbit, so
\[
\bigl|\LTA(m,2)\cdot P\bigr|
=
2^{E(h;S;x_j)}.
\]

\smallskip\noindent
\emph{Disjointness of orbits.}
If two seeds $(h;S;x_j)$ and $(h';S';x_{j'})$ had intersecting orbits,
the gcd of all degree–$r$ monomials in the common codeword would equal
both $h$ and $h'$, hence $h=h'$.  Dividing by $h$ and discarding terms
of degree $<r$ yields two Type–II kernels
$\sum_{f\in S}q_f$ and $\sum_{f'\in S'}q'_{f'}$ with the same orbit.
By the disjointness result for Type–II kernels
\cite[Prop.~4]{rowshan2025weight} we must have $S=S'$ (up to canonical
ordering).  With $h=h'$ and $S=S'$, the remaining difference lies in
the linear term; since $X_j$ and $X_{j'}$ are supported on variables
disjoint from those of $h$ and the $q_i$, equality of polynomials
forces $X_j=X_{j'}$, hence $j=j'$.  Therefore the seeds coincide, and
distinct seeds give disjoint orbits.  Summing $2^{E(h;S;x_j)}$ over all
seeds yields \eqref{eq:nested-rank-l-enum}.
\end{proof}

\end{document}